\documentclass[12pt]{article}

\usepackage{amssymb}
\usepackage{amsmath}
\usepackage{amsthm}
\usepackage{graphicx}
\usepackage{ulem}

\usepackage{color}
\usepackage{bbm}

\newcommand{\mket}[1]{| #1 \rangle}

\newcommand{\mbraket}[2]{\langle #1 | #2 \rangle}
\newcommand{\mketbra}[2]{| #1 \rangle \langle #2 | }
\newcommand{\mdenket}[2]{| #1 \rangle \langle #2 | }

\newcommand{\mtr}[1]{\mathrm{Tr}\left( #1 \right)}
\newcommand{\mptr}[2]{\mathrm{Tr_{#2}}\left( #1 \right)}

\newcommand{\bI}{\mathbb{I}}
\newcommand{\bC}{\mathbb{C}}
\newcommand{\nC}{\mathcal{C}}
\newcommand{\fC}{\mathfrak{C}}
\newcommand{\fD}{\mathfrak{D}}

\newcommand{\fM}{\mathfrak{M}}
\newcommand{\bN}{\mathbb{N}}

\newcommand{\bR}{\mathbb{R}}

\newcommand{\bU}{\mathbb{U}}

\newcommand{\rank}[1]{\mathrm{rank}(#1)}
\newcommand{\pol}{\mathrm{Pol}}
\newcommand{\gen}{\mathrm{Gen}}
\newcommand{\myendfun}{\mathrm{End}}

\newtheorem{observation}{Observation}
\newtheorem{definition}{Definition}
\newtheorem{conclusion}{Conclusion}
\newtheorem{remark}{Remark}
\newtheorem{example}{Example}
\newtheorem{proposition}{Proposition}
\newtheorem{lemma}{Lemma}
\newtheorem{theorem}{Theorem}

\title{A Gramian Approach to Entanglement in Bipartite Finite Dimensional Systems: The case of pure states}

\author{Roman Gielerak, Marek Sawerwain \\ Institute of Control \& Computation Engineering \\ University of Zielona G\'ora \\ e-mail: R.Gielerak@issi.uz.zgora.pl, M.Sawerwain@issi.uz.zgora.pl}

\begin{document}

\maketitle

\begin{abstract}
It has been observed that the reduced density matrices of bipartite qudit pure states possess a Gram matrix structure. This observation has opened a possibility of analysing the entanglement in such systems from the purely geometrical point of view. In particular, a new quantitative measure of an entanglement of the geometrical nature, has been proposed. Using the invented Gram matrix approach, a version of a non-linear purification of mixed states describing the system analysed has been presented.

\end{abstract}

\section{Introduction}

Quantum correlations contained in quantum entangled states describing composite quantum systems are undoubtedly one of the major resources for several quantum information tasks \cite{Schrodinger1935, Nielsen2000, Bengtsson2016, Horodecki2009, Guhne2008}. One of the very interesting features of the quantum entanglement is its monogamous nature which means, roughly, that a total, available amount of quantum correlations contained in quantum states is always of a limited capacity, depending on the very nature of a system considered. In particular, there exist quantum states containing a maximum possible amount of quantum correlations and exactly these states are often most wanted for performing several quantum protocols such as: teleportation of states, implementation of cryptographic protocols and many others. This is the main reason why the mathematical description and the corresponding engineering technologies for the physical preparation of such maximally entangled states seem to be of a great importance \cite{IBM16qEnt}. In the case of composite systems being in the maximally entangled state it is impossible to entangle them further with another quantum system. This monogamy principle \cite{RGielerak2017} is the major element that provides the security of most of the cryptographic protocols implemented technically up to date \cite{Horodecki2010}.

In the case of finite dimensional systems a lot of work has been done on the very nature of quantum entanglement \cite{Schrodinger1935, Nielsen2000, Bengtsson2016, Horodecki2009}. The case of two-partite systems is the best recognised situation. In the case of a two-partite, finite dimensional system being in the pure state, the Schmidt decomposition of the corresponding pure state gives essentially all relevant information on the corresponding quantum correlations. 
The case of many-partite systems and also the case of mixed (even two-partite systems) states are much less recognised despite of many efforts \cite{Nielsen2000, Bengtsson2016, Horodecki2009, Guhne2008, Ritz2018, Caban2017, Bryan2019}.

The case of two-partite systems composed of finite dimensional systems coupled with each other is discussed in the present paper from the perspective of the Gram matrix techniques. In several areas of contemporary research of Gram matrices are extensively used as an important analytic tool. Differential geometry, mathematically oriented statistics problems, quantum chemistry and atomic physics, control theory, machine and deep learning problems are some examples in which Gram matrices are frequently used \cite{Hazenwinkel2001}. As far as Quantum Information Theory is concerned it is hard to show the explicite use of the Gram matrix-based methods. To the best of our knowledge the study presented here constitutes the first systematic and serious attempt at applying Gram matrices theory to analyse the quantum entanglement phenomenon. In particular, a general form of Gram operators describing maximally entangled pure states of such systems is derived and the corresponding amount of entanglement contained in them and defined standardly as the von Neumann entropy of the arising reduced density matrices is calculated.

The paper is organised as follows. In Section~\ref{lbl:sec:mp}, mainly for the reader's convience, we discuss some basic mathematical notions used in the next sections. Section~\ref{lbl:sec:gram:matrix:for:pure:states} presents the basic observations that the reduced density matrices of bipartite quantum systems possess a Gram matrix structure. The first application of this observation is described in Section~\ref{lbl:sec:nonpur} where a non-linear purification map based on the Cholesky decomposition of the corresponding Gram matrices is presented. In Section~\ref{lbl:sec:geo:of:entanglement} new geometrical invariants called Gramian volumes are introduced together with a presentation of their elementary properties including: local $S\bU(d_1) \otimes S \bU(d_2)$ invariance, monotonicity under the local unitary and non-unitary operations such as generalised measurements and CP-transformations (expressed by Completely-Positive Krauss operators) induced by interaction with the environment.

\section{Mathematical preliminaries} \label{lbl:sec:mp}

The notion of frames plays an important role in our analysis of the gramian approach to entanglement of bipartite pure states. In part~\ref{lbl:frames}, we introduce this notion and define other most basic notions which are used in our analysis.  Additionally, in Table~\ref{lbl:tbl:table:of:symbols}, for convenience,  we have collected symbols which are used in subsequent parts of this paper. Other basic definitions and notions, well-known from the theory of linear algebra but related to the study presented in this text, are given in Appendix~A for the reader convienence.

\begin{table}
\caption{Some symbols, notations, sets and functions used in the paper}
\label{lbl:tbl:table:of:symbols}
\begin{center}
\begin{tabular}{c|l}
\hline\hline
Notation & Description \\ \hline\hline
$\bR$ & set of real numbers \\
$\bC$ & set of complex numbers \\
$\bN$ & set of integer numbers \\
$\bU$ & set of unitary operators \\
\hline
$\fC$, $\fD$ & set of vectors \\ \hline
$M$ & set of matrices \\ 
$E_n$, $\bI_n$ & identity matrix \\
$A+B$, $A^{-1}$, & sum of two matrices, inverse of matrix, \\ 
$A^{\star}$, $A^{\dagger}$ & conjugation of matrix, hermitian adjoint of matrix \\ 
$l^{\bot}$ & orthogonal complement of space $l$ \\ \hline
$\otimes$ & Kronecker product of matrices or vectors \\  
$\oplus$ & direct sum of matrices and spaces \\ \hline
$\mbraket{ \cdot }{ \cdot }$ & scalar product \\
$\fM$ & algebra of matrices \\ \hline
$\mtr{A}$ & trace of matrix A \\
$\sigma(A)$ & spectrum of matrix $A$ (counted with multiplicities) \\
$\sigma_{sv}(A)$ & singular values of matrix $A$ (counted with multiplicities) \\ \hline
$E(\bC^d)$ &  set of density matrices on $\bC^d$, i.e. $\rho \geq 0$ and $\mtr{\rho}=1$ \\
$\partial E(\bC^d)$ &  set of pure states on $\bC^d$ \\ \hline
$1:n$ & means the sequence of $1, 2, 3,\ldots, n$ \\
$i \in 1:n$ & means that the index $i$ runs over the sequence $1:n$ \\ \hline
$SL(n, \mathfrak{F})$ & special linear group of degree $n$ over a field $\mathfrak{F}$ \\
$(S)\bU(d)$  & multiplicative group of (special) unitary matrices  acting in $\bC^d$ \\ \hline
$\pi$ & single permutation \\
$S_k$ & symmetric group of permutations \\ 
%... & ... other entries  ... \\
\hline\hline
\end{tabular}
\end{center}
\end{table}

\subsection{Frames} \label{lbl:frames}

For a given $d \in \bN$ (where $\bN$ denotes natural numbers) let $\bC^d$ be a standard, $d$-dimensional Euclidean space over complex numbers denoted as $\bC$~and equipped with the standard scalar product $\mbraket{ \cdot }{ \cdot }$. Any finite set $\Sigma$ of vectors from the space $\bC^d$ will be called a frame. The length of a frame $\Sigma$ is defined as cardinality of the set $\Sigma$ and is denoted as $|\Sigma|$. Rank of $\Sigma$, denoted as $\rank{\Sigma}$, is equal to the dimension of a linear span formed $( =\mathrm{lh}( \Sigma ))$ from the vectors of $\Sigma$. The set of all frames (of a given length $k$, i.e. the finite subsets $\Sigma = \{\psi_1, \ldots, \psi_k\}$ of vectors in $\bC^d$ with $k \geq 1)$ will be denoted as $\mathrm{F}(\bC^d)$ (resp. $\mathrm{kF}(\bC^d)$). The subset of all frames consisting of vector orthogonal to one another vectors is denoted as $\mathrm{OF}(\bC^d)$ and if, moreover, all vectors forming a given frame are normalized, then the frame is called an orthonormal frame and the set of all orthonormal frames in $\bC^d$ is denoted as $\mathrm{ONF}(\bC^d)$. A frame $\Sigma \in \mathrm{F}(\nC^d)$, the linear hull of which is equal to $\bC^d$, and $\mathrm{rank}(\Sigma)=d$ is called a basis. The set of all bases of $\bC^d$ is denoted as $\mathrm{B}(\bC^d)$. In particular, the set of all bases consisting of orthogonal (and normalized) vectors is denoted as $\mathrm{OB}(\bC^d)$ (resp. $\mathrm{ONB}(\bC^d)$).

A canonical basis of $\bC^d$ consists of $d$ vectors $e_i$, for $i=1:d$ and such that the j-th component (of the column as written) $e_i$ is equal ${(e_i )}_j =\delta_{ij}$  where $\delta_{ij}$ is a discrete Kronecker delta symbol. Any vector $\Psi \in \bC^d$ can be in a unique way written as:
\begin{equation}
\Psi = \sum_{i=1}^{d} c_i e_i ,
\end{equation}
where $c_i = \mbraket{ e_i }{ \Psi }_{\bC^d}$.

\section{Gram matrix description of pure states} \label{lbl:sec:gram:matrix:for:pure:states}

Let $d_1$, $d_2$ be given integers and let $\Psi$ be a vector in the space $\bC^{d_1} \otimes \bC^{d_2}$. Using the canonical bases $(e_i)$ and $(f_j)$ in the corresponding spaces $\bC^{d_1}$, $\bC^{d_2}$ resp., the vector $\Psi$ can be expanded as:

\begin{eqnarray}
\Psi = \sum_{i,j} c_{ij} e_i \otimes f_j, \; c_{ij} = \mbraket{e_i \otimes f_j}{ \Psi } \; \mathrm{for} \; i=1:d_1, j=1:d_2 .
\end{eqnarray}

Let $\{ E^{\otimes}_{i}, i=1:d_1 d_2 \}$ be the canonical basis of the product space $\bC^{d_1} \otimes \bC^{d_2} \cong \bC^{d_1 d_2}$. Then the vector $\Psi$ can also be decomposed as:
\begin{eqnarray}
\Psi = \sum_{i=1:d_1 d_2} D_{i} E_i^{\otimes}, \; D_{i} = \mbraket{E_i^{\otimes}}{ \Psi } \; \mathrm{for} \; i=1:d_1 d_2 .
\end{eqnarray}

Let $\Psi \in \bC^{d_1} \otimes \bC^{d_2}$, then we can write:
\begin{gather}
\Psi = \sum_{i,j} c_{ij} e_i \otimes f_j = \sum_{i=1:d_1} e_i \otimes \left(\sum_{j=1:d_2} c_{ij} f_j \right) = \sum_{j=1:d_2} \left( \sum_{i=1:d_1} c_{ij} e_i \right) \otimes f_j .
\end{gather}

Using this we can define two frames:
\begin{equation}
\Gamma^R(\Psi) = \{ \psi^{R}_{i} = \sum_{j=1:d_2} c_{ij}f_j, i=1:d_1 \},
\end{equation}
in $\bC^{d_2}$ and of length $d_1$, and
\begin{equation}
\Gamma^L(\Psi) = \{ \psi^{L}_{i} = \sum_{i=1:d_1} c_{ij}e_i, j=1:d_2 \},
\end{equation}
in $\bC^{d_1}$ and of length $d_2$.

In terms of the introduced frames we can write:
\begin{equation}
\Psi = \sum_{i=1:d_1} e_i \otimes \psi^{R}_{i},
\end{equation}
or
\begin{equation}
\Psi = \sum_{j=1:d_2} \Psi^{L}_{j}  \otimes f_j.
\end{equation}

Let us define the following maps:
\begin{equation}
\mathrm{J}^R : \bC^{d_1} \rightarrow \bC^{d_2}, \; \mathrm{for} \; v=\sum_{i=1:d_1} v_i e_i, \; \mathrm{J}^{R}(v) = \sum_{i=1:d_1} v_i \psi^{R}_i ,
\end{equation}
and similarly,
\begin{equation}
\mathrm{J}^L : \bC^{d_2} \rightarrow \bC^{d_1}, \; \mathrm{for} \; w=\sum_{j=1:d_2} w_j f_j, \; \mathrm{J}^{L}(w) = \sum_{j=1:d_2} w_j \psi^{L}_j .
\end{equation}

The adjoint operation to $\mathrm{J}^R$, the operation $(\mathrm{J}^R)^{\dagger}$ is given as:
\begin{equation}
(\mathrm{J}^R)^{\dagger} : \bC^{d_2} \rightarrow \bC^{d_1}, \; \mbraket{w}{ \mathrm{J}^R v } = \mbraket{(\mathrm{J}^R)^{\dagger} w}{v},
\end{equation}
from which it follows
\begin{equation}
{(\mathrm{J}^R)^{\dagger}}_{ij} = \mbraket{(\mathrm{J}^R)^{\dagger} f_j}{ e_i } = \mbraket{f_j}{\psi^R_i} = c_{ij} .
\end{equation}

Similarly, we can compute the matrix of the operator $(\mathrm{J}^L)^{\dagger}$
\begin{equation}
{(\mathrm{J}^L)^{\dagger}}_{ij} = \mbraket{(\mathrm{J}^L)^{\dagger} e_j}{ f_i } = \mbraket{e_j}{\psi^L_i} = c_{ji} .
\end{equation}

\begin{definition}
Let $\Psi \in \bC^{d_1} \otimes \bC^{d_2}$. The right Gram operator of $\Psi$ is defined as:
\begin{equation}
\Delta^R( \Psi ) = (\mathrm{J}^R)^{\dagger} \circ \mathrm{J}^R : \bC^{d_1} \rightarrow \bC^{d_1} .
\end{equation}
Similarly, the left Gram operator of the vector $\Psi$ is defined as:
\begin{equation}
\Delta^L( \Psi ) = (\mathrm{J}^L)^{\dagger} \circ \mathrm{J}^L : \bC^{d_2} \rightarrow \bC^{d_2} .
\end{equation}
And finally, the Gram operator of $\Psi$ is defined as:
\begin{equation}
\Delta(\Psi) = \Delta^R(\Psi) \otimes \Delta^L(\Psi) : \bC^{d_1} \otimes \bC^{d_2} \rightarrow \bC^{d_1} \otimes \bC^{d_2} .
\end{equation}
$\square$
\end{definition}

\begin{remark}
Let $\Psi \in \bC^{d_1} \otimes \bC^{d_2}$,
\begin{equation}
\Psi = \sum_{i,j} c_{ij} e_i \otimes f_j,
\end{equation}
be given. Then the density matrix $Q(\Psi)$ of the pure state $\mket{\Psi}$ (where the well known bra and ket notation is used) is given by the following formula:
\begin{equation}
Q(\Psi) = \mketbra{\Psi}{\Psi} = \sum_{\alpha\beta}\sum_{\alpha'\beta'} \overline{c_{\alpha\beta}} c_{\alpha'\beta'} \mketbra{e_{\alpha} \otimes f_{\beta}}{e_{\alpha'} \otimes f_{\beta'}} .
\end{equation}
Computing the corresponding reduced density matrices, we obtain:
\begin{equation}
Q^2(\Psi) = \mptr{Q(\Psi)}{\bC^{d_1}}, \; \; \;Q^1(\Psi) = \mptr{Q(\Psi)}{\bC^{d_2}} .
\end{equation}
It follows that $Q^1(\Psi) = \Delta^{R}(\Psi)$ and $Q^2(\Psi) = \Delta^{L}(\Psi)$. 

Therefore, we conclude that the corresponding Gram matrices $\Delta^{R}(\Psi)$ and $\Delta^{L}(\Psi)$ have an important, physical meaning and as such they are physically observable quantities (for more details see \cite{Guhne2008}). $\square$
\end{remark}

\begin{example}
Let us consider the case $d_1 = d_2 = 2$ in more detail. Let $\Psi \in \bC^2 \otimes \bC^2 \cong \bC^4$, $\Psi = \sum_{\alpha=1:4} c_{\alpha} E^{\otimes}_{\alpha}$ where $(E^{\otimes}_{\alpha})_{\beta} = \delta_{\alpha\beta}$ is the canonical basis of the space $\bC^4$. The R-frame of $\Psi$ is easy to compute:
\begin{eqnarray}
\mathrm{FR}(\Psi) & = & \left( \psi_1^R = \left( \begin{array}{c} c_1 \\ c_2 \end{array}\right), \psi_2^R = \left( \begin{array}{c} c_3 \\ c_4 \end{array}\right) \right), \notag \\
\mathrm{FL}(\Psi) & = & \left( \psi_1^L = \left( \begin{array}{c} c_1 \\ c_3 \end{array}\right), \psi_2^L = \left( \begin{array}{c} c_2 \\ c_4 \end{array}\right) \right) .
\end{eqnarray}

The corresponding Gram matrices are given as
\begin{eqnarray}
\Delta^{R}( \Psi ) & = & \left( \begin{array}{cc}
|c_1|^2 + |c_2|^2 						&  \overline{c_1} c_3 + \overline{c_2} c_4 \\
\overline{c_3} c_1 + \overline{c_4} c_2 & |c_3|^2 + |c_4|^2
\end{array} \right) , \notag \\ 
\Delta^{L}( \Psi ) & = & \left( \begin{array}{cc}
|c_1|^2 + |c_3|^2 						&  \overline{c_1} c_2 + \overline{c_3} c_4 \\
\overline{c_2} c_1 + \overline{c_4} c_3 & |c_2|^2 + |c_4|^2
\end{array} \right) .
\end{eqnarray}

Denoting:
\begin{eqnarray}
A & = & |c_1|^2 + |c_2|^2 , \notag \\ 
B & = & |c_3|^2 + |c_4|^2 , \notag \\ 
C & = & |c_1|^2 + |c_3|^2 , \notag \\ 
D & = & |c_2|^2 + |c_4|^2 , \notag \\ 
C_{13} & = & \overline{c_1} c_3 + \overline{c_2} c_4  , \notag \\ 
C_{12} & = & \overline{c_1} c_2 + \overline{c_3} c_4 , \notag \\ 
C_{31} & = & \overline{c_3} c_1 + \overline{c_4} c_2 , \notag \\ 
C_{21} & = & \overline{c_2} c_1 + \overline{c_4} c_3 , \notag 
\end{eqnarray}
we obtain the explicite formula for the full Gram matrix of the vector $\Psi$:
\begin{equation}
\Delta(\Psi) = \left(
\begin{array}{cccc}
AC      	 & AC_{12} 		& CC_{13}      & C_{13}C_{12} \\
AC_{21} 	 & AD      		& C_{13}C_{21} & DC_{13} \\
C_{31}C 	 & C_{31}C_{12} & BC		   & BC_{12} \\
C_{31}C_{21} & DC_{31}      & BC_{21} 	   & BD
\end{array}
\right) .
\end{equation}
$\square$
\end{example}

\begin{example}
Let $\Psi = \psi_1 \otimes \psi_2$, where $\psi_1 = \sum_{i=1:d_1} c_i e_i$, $\psi_2 = \sum_{i=1:d_2} d_i f_i$ be a separable vector in $\bC^{d_1} \otimes \bC^{d_2}$. Then by an easy computation we get:
\begin{eqnarray}
RF( \Psi ) & = & \{ \psi^{R}_{\alpha} = c_{\alpha} \psi_2 \; \mathrm{for} \; \alpha=1:d_1  \}, \notag \\
LF( \Psi ) & = & \{ \psi^{L}_{\beta} = d_{\beta} \psi_1 \; \mathrm{for} \; \beta=1:d_2  \}.
\end{eqnarray}
Therefore,
\begin{eqnarray}
{\Delta^{R}( \Psi )}_{\alpha\beta} = \mbraket{\psi^R_{\alpha}}{\psi^R_{\beta}} = \overline{c_{\alpha}}c_{\beta} {|| \psi_2 ||}^2 , \notag \\
{\Delta^{L}( \Psi )}_{\alpha\beta} = \mbraket{\psi^L_{\alpha}}{\psi^L_{\beta}} = \overline{d_{\alpha}}d_{\beta} {|| \psi_1 ||}^2 .
\end{eqnarray}
Defining the following vectors $\fC=[c_1, \ldots, c_{d_1} ]$ and $\fD=[b_1, \ldots, b_{d_2} ]$ and multiplying them as matrices we have the following equalities:
\begin{equation}
\Delta^R(\Psi) = \fC^{\dagger} \fC \; \mathrm{and} \; \Delta^L(\Psi) = \fD^{\dagger} \fD .
\end{equation}
$\square$
\end{example}

It is not difficult to note the following:
\begin{proposition}
A vector $\Psi \in \bC^{d_1} \otimes \bC^{d_2}$ is a separable vector iff there exist two vectors (as rows) $\fC \in \bC^{d_1}$ and $\fD \in \bC^{d_2}$ such that $\Delta^R(\Psi) = \fC^{\dagger} \fC$ and $\Delta^L(\Psi) = \fD^{\dagger} \fD$.
\end{proposition}

\begin{proposition}
Let $\Psi \in \bC^{d_1} \otimes \bC^{d_2}$. Then, the corresponding Gram operators $\Delta^R$, $\Delta^L$ and $\Delta$ have the following matrix representations in the
canonical bases $(e_i)$, $(f_j)$ and $(E^{\otimes}_{k})$, respectively:
\begin{eqnarray}
{\Delta^{R}( \Psi )}_{ij} & = & \mbraket{\psi^R_{i}}{\psi^R_{j}}_{\bC^{d_1}} \; \mathrm{for} \; i,j = 1:d_1, \notag \\
{\Delta^{L}( \Psi )}_{ij} & = & \mbraket{\psi^L_{i}}{\psi^L_{j}}_{\bC^{d_2}} \; \mathrm{for} \; i,j = 1:d_2. 
\label{lbl:eq:delta:gram:oper}
\end{eqnarray}
and
\begin{equation}
\Delta(\Psi)_{ij} = \mbraket{E^{\otimes}_j}{\Delta(\Psi)E^{\otimes}_i} = \mbraket{e_{\beta_1} \otimes f_{\alpha_1} }{\Delta(\Psi)_{e_{\beta_1}} \otimes f_{\alpha_1}} = {\Delta^{R}( \Psi )}_{\beta_1\beta_2} {\Delta^{L}( \Psi )}_{\alpha_{1}\alpha_{2}},
\end{equation}
for

\begin{equation}
i=(\alpha_1 -1 ) d_1 + \beta_1, \; j=(\alpha_2 - 1) d_1 + \beta_2, 
\end{equation}
where $1 \leq \alpha \leq d_2$ and $0 \leq \beta \leq d_1 - 1$. See Eq.~(\ref{lbl:eq:alpha:i:index}) and Eq.~(\ref{lbl:eq:alpha:ij:index}).

\end{proposition}

\begin{proposition}
Let $\Psi \in \bC^{d_1} \otimes \bC^{d_2}$. Then all the introduced Gram operators connected with $\Psi$ are hermitian and positive semi-definite ,
which means that for any finite sequence $\alpha_i$, $i=1:d_1$ of complex numbers the following inequality holds:
\begin{equation}
\sum_{i,j=1:d_1} \alpha_i \overline{\alpha_j} {\Delta^R( \Psi )}_{ij} \geq 0.
\label{lbl:eq:gram:pos:semidefinite}
\end{equation}
The minimal value equal to zero in Eq.~(\ref{lbl:eq:gram:pos:semidefinite}) is attained , i.e $\dim \mathrm{Ker}(\Delta^{R}( \Psi )) > 0$, iff the frame $RF(\Psi)$ has rank less than $d_1$ , i.e the vectors $\Psi^R_i$ forming $RF(\Psi)$ are linearly dependent. 
And similarly for the case of $\Delta^L(\Psi)$ and $\Delta(\Psi)$.
\end{proposition}

\begin{proof}
Using the explicite form of the matrix elements of $\Gamma^R$ as given by (\ref{lbl:eq:delta:gram:oper}) it follows:
\begin{gather}
\sum_{i,j=1:d_1} \alpha_i \overline{\alpha_j} {\Delta^R( \Psi )}_{ij} = \sum_{i,j=1:d_1} \alpha_i \overline{\alpha_j} {\mbraket{\psi^R_i}{\psi^R_j}}_{\bC^{d_1}} = \notag \\ 
{{ \left| \mbraket{\sum_{i=1:d_1} \alpha_i \psi^R_i}{\sum_{i=1:d_1} \alpha_i \psi^R_i} \right|}^2} = {\sum_{i=1:d_1} \alpha_i \psi^R_i  }^2 \geq 0 .
\end{gather}
And similarly for the remaining cases. Non-negativity of the Gram operator $\Delta(\Psi)$ follows also from the fact that the tensor product of positive semi-definite (positive definite) matrices is also a positive semi-definite (resp. positive definite) matrix.
\end{proof}

\begin{proposition}
Let $\Psi \in \bC^{d_1} \otimes \bC^{d_2}$. Then the following formula holds:
\begin{equation}
\mtr{\Delta^R ( \Psi ) } = \mtr{ \Delta^L ( \Psi ) }= {|| \Psi ||}^2 ,
\end{equation}
and
\begin{equation}
\mtr{ \Delta ( \Psi ) } =  {|| \Psi ||}^4 .
\label{lbl:eq:trace:to:four}
\end{equation}
\end{proposition}

\begin{proof}
From Eq.~(\ref{lbl:eq:delta:gram:oper}) it follows:
\begin{eqnarray}
\mtr{\Delta^R(\Psi)} = \sum_{i=1:d_1} {\Delta^R(\Psi)}_{ii} = \sum_{i=1:d_1} {|| \psi^R_i ||}^2 = {|| \Psi ||}^2.
\end{eqnarray}
Equation~(\ref{lbl:eq:trace:to:four}) follows from the fact that the trace of tensor product of operators is equal to the product of traces, see TP2 (ii) property (Appendix~\ref{lbl:subapp:Kronecker:Product}). In particular cases of normalized vectors, the corresponding traces are all equal to one.
\end{proof}

Let $(S)\bU(d)$ stand for the multiplicative group of (special) unitary transformations of the space $\bC^d$. Then we have the following observation.

\begin{proposition}
Let $\Psi \in \bC^{d_1} \otimes \bC^{d_2}$. Then the Gram operators of $\Psi$ as defined above obey the following invariance properties:
\begin{itemize}
\item[(1)] for any  $U \in S\bU(d_2)$:
\begin{equation}
\Delta^R((1 \otimes U)\Psi) = \Delta^R(\Psi),
\end{equation}
\item[(2)] for any  $U \in S\bU(d_1)$:
\begin{equation}
\Delta^L((U \otimes 1)\Psi) = \Delta^L(\Psi),
\end{equation}
\item[(3)] for any  $U_1 \in S\bU(d_1)$ and for any $U_2 \in S\bU(d_2)$:
\begin{equation}
\Delta((U_1 \otimes U_2)\Psi) = \Delta(\Psi).
\end{equation}
\end{itemize}
\end{proposition}

\begin{proof}
It follows from the very definitions.
\end{proof}

Let us start with the following presumably well-known and intuitively obvious observation.

\begin{lemma}
Let $F_1=(v_1, \ldots, v_k)$, $F_2=(w_1, \ldots, w_k)$ be two k-frames in $\bC^d$ and let $\Delta(F_1)$, resp. $\Delta(F_2)$ be the corresponding Gram matrix of $F_1$, resp. of $F_2$. Assume that $\Delta(F_1)=\Delta(F_2)$, then there exists a unitary map $U \in S\bU(d)$ such that $F_2 = U F_1$.
\label{lbl:lemma:unitary:map:for:frames}
\end{lemma}

\begin{proof}
We divide the proof into four cases:
\begin{itemize}
\item[(Case 1a:)] $k=d$ and $\rank{\Delta(F_1)} = d$ . \\
Let us start with the case $k=d$ and $d=\rank{F_1}=\rank{F_2}$. From the equality $\Delta(F_1) = \Delta(F_2)$ it follows that the angles between pairs of vector forming $F_1$ and $F_2$ are the same. Let us imagine as a basis of $\bC^d$ the coordinate systems with axes defined as a  set  of  vectors forming $F_1$. Let us imagine the same also for the coordinate system formed from vectors of $F_2$. Both systems can be seen as skeletons of some rigid body, which in both cases looks like (up to some extent of course) a rigid hedgehog. From the elementary arguments it follows that one can rigidly rotate the coordinate system in $\bC^d$ formed by $F_2$ into such positions that the corresponding coordinate axes will coincide with these of $F_1$. But any rotation in $\bC^d$ is an element of $S\bU(d)$ group.

\item[(Case 1b:)] $k=d$ and $\rank{F_1} = m < d$ . \\
Let $lh(F_1)$ be a subspace of dimension $m$ in $\bC^d$ formed by a subframe $F_1^{'} \subset F_1$, $F_1^{'}=\{ v_{i_1}, \ldots, v_{i_m} \}$ and let $F_2^{'}=\{ w_{i_1}, \ldots, w_{i_m} \} \subset F_2$ be the corresponding k-subframe of $F_2$. From the equality $\Delta(F_1)=\Delta(F_2)$ it follows that $\Delta(F_1^{'})=\Delta(F_2^{'})$ and we can apply the argument used in Case 1a and conclude that there exists a unitary map:
\begin{equation}
\begin{array}{rll}
U^{'} 			& : & lh(F^{'}_1) \rightarrow lh(F^{'}_{2}) \\
\mathrm{with}   &   & U^{'}(v_{i_{\alpha}}) = w_{i_{\alpha}}, \alpha=1:k .
\end{array}
\end{equation}
But
\begin{equation}
\bC^d = lh(F^{'}_{1}) \oplus lh(F^{'}_1)^{\bot},
\end{equation}
where $( \cdot )^{\bot}$ means the orthogonal complement of $( \cdot )$ and we can define $U^{\bot}$:
\begin{equation}
{U^{'}}^{\bot} : lh(F^{'}_1)^{\bot} \rightarrow lh(F^{'}_2)^{\bot} .
\end{equation}
Then the map $U=U^{'} \oplus {U^{'}}^{\bot}$ is the desired unitary map U: $\bC^d \rightarrow \bC^d$ and such that $U F_1 = F_2$.

\item [(Case 2:)] $k<d$. \\
Similar arguments can be used as in (Case 1b).

\item [(Case 3:)] $k>d$. \\
Let $k^{'} = \rank{\Delta(F_1)} = \rank{\Delta(F_2)}$ and let $F_1^{'}=\{ v_{i_1}, \ldots, v_{i_k} \} \subset F_1$ be such that $\rank{F^{'}_{1})=k^{'}}$, i.e. $lh(F^{'}_{1}) = lh(F_1)$ and similarly for $F_2$. Let $F^{'}_{2}$ be the spanning $k^{'}$-subframe chosen from $F_2$. From
\begin{equation}
\Delta(F_2^{'}) = \Delta(F_1^{'}),
\end{equation}
it follows that there exists a unitary map $U_1 : lh(F_1^{'}) \rightarrow lh(F_2^{'})$ such that $U_1(v_{i^{'}}) = w_i$ for $v_i^{'} \in F_1$ from which it follows that in fact $U_1(v_i)=w_i$ for all $i=1:k$.

Decomposing:
\begin{equation}
\bC^d = lh(F^{'}_1) \oplus lh(F^{'}_{1})^{\bot} =  lh(F^{'}_2) \oplus lh(F^{'}_{2})^{\bot},
\end{equation}
and taking any unitary map $U^{\bot}$:
\begin{equation}
U^{\bot}_1 : lh(F^{'}_1)^{\bot} \rightarrow lh(F^{'}_{2})^{\bot},
\end{equation}
and then we have:
\begin{equation}
U = U_1 \oplus U_2^{\bot},
\end{equation}
\end{itemize}
and this concludes the proof of Lemma~\ref{lbl:lemma:unitary:map:for:frames}.
\end{proof}

\begin{example}
Let us consider the asymmetric case $\bC^2 \otimes \bC^3$ in more detail. Let
\begin{displaymath}
\Psi = \sum_{\alpha=1}^{2} \sum_{\beta=1}^{3} \psi_{\alpha\beta}  e_{\alpha} \otimes f_\beta ,
\end{displaymath}
where $\{ e_1, e_2 \}$, resp. $\{\ f_1, f_2, f_3 \}$ are the canonical bases in $\bC^2$, resp. in $\bC^3$. Writing
\begin{displaymath}
\Psi = \sum^2_{\alpha=1} e_{\alpha} \otimes r^{\alpha},
\end{displaymath}
where 
\begin{displaymath}
r^{\alpha} = \sum_{\beta=1}^{3} \psi_{\alpha\beta} f_{\beta} \in \bC^{3} .
\end{displaymath}
We can define
\begin{displaymath}
\begin{array}{lll}
j_R & : & \bC^2 \rightarrow \bC^3 \\
	&	& e_{\alpha}  \rightarrow r_{\alpha}
\end{array}.
\end{displaymath}
By extending by linearity to the whole space $\bC^2$ it is not difficult to compute
\begin{displaymath}
\Delta_{R}(\Psi) = \left[
\begin{array}{cc}
\mbraket{r_1}{r_1}_{\bC^3} & \mbraket{r_2}{r_1}_{\bC^3} \\
\mbraket{r_1}{r_2}_{\bC^3} & \mbraket{r_2}{r_2}_{\bC^3} \\
\end{array}
\right] .
\end{displaymath}
Similarly,
\begin{displaymath}
\Delta_{L}(\Psi) = \left[
\begin{array}{ccc}
\mbraket{l_1}{l_1}_{\bC^2} & \mbraket{l_2}{l_1}_{\bC^2} & \mbraket{l_3}{l_1}_{\bC^2} \\
\mbraket{l_1}{l_2}_{\bC^2} & \mbraket{l_2}{l_2}_{\bC^2} & \mbraket{l_3}{l_2}_{\bC^2} \\
\mbraket{l_1}{l_3}_{\bC^2} & \mbraket{l_2}{l_3}_{\bC^2} & \mbraket{l_3}{l_3}_{\bC^2} \\
\end{array}
\right],
\end{displaymath}
where
\begin{equation}
\Psi = \sum_{\beta=1}^{3} l_\beta \otimes f_{\beta},
\end{equation}
and
\begin{equation}
l_{\beta} = \sum_{\alpha=1}^{2} \psi_{\alpha\beta} e_{\alpha} \in \bC^2.
\end{equation}
The following facts hold true:
\begin{itemize}
\item[(1)] $\mathrm{rank} \{ l_1, l_2, l_3 \} \leq 2$,
\item[(2)] $\det( \Delta_{l}(\Psi) ) = 0$,
\item[(3)] $\sigma\left( \Delta_{L}(\Psi)) \setminus \sigma(\Delta_{R}(\Psi) \right) = \{ 0 \}$ .
\end{itemize}
\label{lbl:eq:example:23:qsys}
\end{example}
\begin{proof}
Proof of fact (3) in Example~\ref{lbl:eq:example:23:qsys}.

Let 
\begin{equation}
\begin{array}{lll}
\pi_{12} & : & \bC^2  \otimes \bC^{3} \rightarrow \bC^3 \otimes \bC^{2} , \\
		 &	 & \psi_1 \otimes \psi_2  \rightarrow \psi_2 \otimes \psi_1 ,
\end{array}
\end{equation}
and let it extend by linearity to the whole space $\bC^2 \otimes \bC^3$ . The map $\pi_{12}$ is bijective and preserves the scalar product, therefore, $\pi_{12}$ is a unitary isomorphism in between the corresponding spaces. Let $\Psi^{\pi} = \pi_{12}(\Psi) \in \bC^3 \otimes \bC^2$, then $\Delta^{R}(\Psi^{\pi})=\Delta^{L}(\Psi)$.
Let
\begin{equation}
\Psi = \sum_{\alpha=1}^{2} s_{\alpha} \overline{e_{\alpha}} \otimes \overline{f_{\alpha}},
\end{equation}
be a canonical Schmidt decomposition of $\Psi$. Let us consider:
\begin{displaymath}
\Psi^{'} = \sum_{\beta=1}^{3} s^{'}_{\beta} \overline{f_{\beta}} \otimes \overline{e_{\beta}},
\end{displaymath}
where $s^{'}_{\beta} = s_{\alpha}$ for $\beta \leq 2$ and $s^{'}_{3} = 0$.

From the uniqueness (modulo kernel of the corresponding Gram operator $\Delta$, see \cite{Nielsen2000, Bengtsson2016}) it follows that
\begin{equation}
\Psi^{'} = \Psi^{\pi},
\end{equation}
and the point (3) above follows.
\end{proof}

Let us consider a general asymmetric case $d_1 < d_2$. Proceeding exactly as in Example~\ref{lbl:eq:example:23:qsys} we can prove the following proposition.

\begin{proposition}
Let $\Psi \in \bC^{d_1} \otimes \bC^{d_2}$ with $d_1 < d_2$ and let $\Delta_{R}(\Psi)$, resp. $\Delta_{L}(\Psi)$ be the corresponding Gram matrices of $\Psi$. Then:
\begin{itemize}
\item[(1)] $\mathrm{rank}(\Delta_{L}) \leq \mathrm{rank}(\Delta_{R})$, in fact $\rank{\Delta_L} = \rank{\Delta_R}$,
\item[(2)] $\det( \Delta_{L} ) = 0$,
\item[(3)] $\sigma(\Delta_{L}(\Psi)) \setminus \sigma(\Delta_{R}(\Psi)) = \{ 0, \ldots, 0\}$ (with $d_2 - d_1$ zeros at least\footnote{As explained before the notion of spectrum $\sigma(A)$ of a matrix $A$ is defined as some quasi(multi)-set structure with multi elements (eigenvalues listed together with their multiplicities).}).
\end{itemize}
\end{proposition}

Now, we are ready to formulate the following result:

\begin{proposition}
Let $\psi_1, \psi_2 \in \bC^{d_1} \otimes \bC^{d_2}$ and let us assume that:
\begin{itemize}
\item[(a)] Let $\Delta^R(\psi_1) = \Delta^R(\psi_2)$ and $\rank{\Delta^{R}(\psi_1)}=d_2$. Then there exists a unique $U \in S\bU(d_2)$ such that $\psi_1 = (1 \otimes U) \psi_2$. If $\rank{\Delta^{R}(\psi_1)} \neq d_2$, then the uniqueness part is not valid in general.

\item[(b)] Let $\Delta^L(\psi_1) = \Delta^L(\psi_2)$ and $\rank{\Delta^{L}(\psi_1)}=d_1$. Then there exists a unique $U \in S\bU(d_1)$ such that $\psi_1 = (U \otimes 1) \psi_2$. If $\rank{\Delta^{L}(\psi_1)} \neq d_1$, then the uniqueness part is not valid in general.

\item[(c)] Let $\Delta(\psi_1) = \Delta(\psi_2)$ and let $\rank{\Delta(\psi_1)}=d_1 \cdot d_2$. Then there exists a unique pair $U_1 \in S\bU(d_2)$, $U_2 \in S\bU(d_1)$ such that $\psi_1 = (U_2 \otimes U_1) \psi_2$. If $\rank{\Delta(\psi_1)} \neq d_1 \cdot d_2$, then the uniqueness part is no longer valid in general.
\end{itemize}
\end{proposition}

\subsection{Relative Gram operators}

Let $\Psi \in \bC^{d_1} \otimes \bC^{d_2}$, $\Phi \in \bC^{d_1} \otimes \bC^{d_2}$ and let $\Psi^R = (\psi_1, \ldots, \psi_{d_1})$ and $\Phi^R = (\phi_1, \ldots, \phi_{d_1})$ be R-frames of $\Psi$, resp. of $\Phi$.

For $\alpha, \beta \in \bC$ we define a new vector $\Theta = \alpha \Psi + \beta \Phi \in \bC^{d_1} \otimes \bC^{d_2}$ and we have
\begin{equation}
\mathrm{RF}(\alpha \Psi + \beta \Phi) =  \alpha \Psi^R + \beta \Phi^R .
\end{equation}

Therefore, using (\ref{lbl:eq:delta:gram:oper}) we have the following formula:
\begin{equation}
\Delta^R(\alpha \Psi + \beta \Phi) = {|\alpha|}^2 \Delta^R(\Psi) + \overline{\alpha} \beta \Delta^R (\Psi|\Phi) + \alpha \overline{\beta} \Delta^R (\Phi|\Psi) + {|\beta|}^2 \Delta^R(\Psi),
\end{equation}
where
\begin{equation}
{\Delta^R( \Psi | \Phi)}_{ij} = {\mbraket{\phi_j}{\psi_i}}_{\bC^d} .
\end{equation}

It is clear that an identical definition works in the case of left frames also.

\begin{definition}
Let $\Psi = (\psi_1, \ldots, \psi_{k})$ and $\Phi = (\phi_1, \ldots, \phi_{k})$ be two k-frames %of $\Psi$, resp. of $\Phi$ 
in the space $\bC^{d}$. Then we define the relative Gram operators (equivalently Gram matrices) of them as:
\begin{equation}
{\Delta(\Psi | \Phi)}_{\alpha\beta} = {\mbraket{\phi_{\beta}}{\psi_{\alpha}}}_{\bC^d} .
\end{equation}
\end{definition}

The following elementary properties of the introduced bilinear functional $\Delta(\cdot | \cdot)$ on the space of k-frames are listed here:
\begin{itemize}
\item[RG(1)] If $\Psi=\Phi$, then $\Delta(\Psi|\Phi) = \Delta(\Psi)$.
\item[RG(2)] For any $\Psi, \Phi \in \mathrm{kF}(\bC^{d})$:
\begin{equation}
{\Delta(\Psi | \Phi)}^{\dagger} = \Delta( \Phi | \Psi ) .
\end{equation}
\item[RG(3)] The self-adjoint part of $\Delta$, denoted as $s\Delta$, is given by:
\begin{equation}
s\Delta(\Psi | \Phi) =  s\Delta( \Psi | \Phi) ) + s\Delta( \Phi | \Psi ) ,
\end{equation}
and then
\begin{equation}
{s\Delta(\Psi | \Phi)}^{\dagger} = s\Delta( \Psi | \Phi ) .
\end{equation}
\item[RG(4)] For any $\alpha \in \bC$, any $\Psi \in k\mathrm{F}(\bC^d)$:
\begin{equation}
\Delta( \alpha \Psi ) = {|\alpha|}^2 \Delta( \Psi ) .
\end{equation}
\item[RG(5)] For any $\Psi, \Phi \in k\mathrm{F}(\bC^d)$:
\begin{equation}
\Delta(\Psi + \Phi) = \Delta(\Psi) + \Delta(\Phi)  + \Delta(\Psi | \Phi) + \Delta(\Phi | \Psi) .
\end{equation}
\item[RG(6)] If $s\Delta(\Psi | \Phi) \equiv  \Delta( \Psi | \Phi) + \Delta(\Phi|\Psi)$,  then $\Delta(\Psi) + \Delta(\Phi) \geq -s\Delta(\Psi | \Phi)$.
\item[RG(7)] For any $U \in S\bU(k)$ we have $\Delta( U\Psi | U\Phi ) = \Delta( \Psi | \Phi)$.
\end{itemize}

Using the notion of the graded Grassmann algebra construction $\Lambda(\bC^d)$ together with the corresponding exterior, antisymmetric cross product $\Lambda$ the following result can be proved \cite{Bourbaki1989, MacLane1999, Spivak1965}.

\begin{proposition}
Let $\Psi=(\psi_1, \ldots, \psi_k) \in kF(\bC^d)$, $\Phi=(\phi_1, \ldots, \phi_k) \in kF(\bC^d)$. Then
\begin{equation}
\det ( \Delta ( \Psi | \Phi ) ) = \mbraket{ \psi_1 \wedge \ldots \wedge \psi_k }{ \phi_1 \wedge \ldots \wedge \phi_k }_{\Lambda^k(\bC^d)} .
\end{equation}
In particular:
\begin{itemize}
\item[(i)] $\det ( \Delta ( \Psi ) ) = { || \psi_1 \wedge \ldots \wedge \psi_k ||}^2_{\Lambda^k(\bC^d)}$,
\item[(ii)] for any $k > d$, $\det ( \Delta( \Psi | \Phi ) ) = 0$.
\end{itemize}
\end{proposition}

\section{Non-linear purification} \label{lbl:sec:nonpur}

Let $T_{u(l)}(\bC^d)$ stand for the set of upper (resp. lower) triangular matrices of size $d \times d$. The following results are evident. 

\begin{lemma}
Basic properties of triangular matrices
\begin{itemize}
\item[1.] Let $A,B \in T_{u}(\bC^d)$, then
\begin{itemize}
\item[(i)] $A+B \in T_{u}(\bC^d)$,
\item[(ii)] $A \cdot B \in T_{u}(\bC^d)$,
\item[(iii)] if $A^{-1}$ exists, then $A^{-1} \in T_{u}(\bC^d)$,
\item[(iv)] $\forall_{c \in \bC} \;\;\; c \cdot A \in T_{u}(\bC^d)$.
\end{itemize}
\item[2.] The same results are valid if  $A,B \in T_{l}(\bC^d)$.
\end{itemize}
\end{lemma}

\begin{lemma} Tensor products of triangular matrices
\begin{itemize}
\item[1.] Let $A \in T_{u}(\bC^d)$ and $B \in T_{u}(\bC^{d^{'}})$ then $A \otimes B \in T_{u}(\bC^{d \cdot d^{'}})$.
\item[2.] The same is valid if $A \in T_{l}(\bC^{d})$ and $B \in T_{l}(\bC^{d^{'}})$.
\end{itemize}
\label{lbl:lemma:tensor:products}
\end{lemma}
\begin{proof}
The proof follows from the very definitions and notions presented in Appendix~\ref{lbl:ssec:tensor:matrix}.
\end{proof}

\begin{lemma} An algebra of upper and lower matrices
\begin{itemize}
\item[(1)] The sets $T_{u(l)} (\bC^d)$ form self-adjoint algebras of the algebra $\fM(\bC^d)$.
\item[(2)] The sets $T_{u(l)} (\bC^d)$, from the point of view of the Lie algebra theory, form solvable Lie subalgebras of the Lie algebra $gl(\bC^d)$ of all $d \times d$ matrices.
\end{itemize}
\end{lemma}
\begin{proof}
Point (1) is obvious. For the proof of (2) we refer the reader to \cite{LieBook1, LieBook2}.
\end{proof}

Let $A$ be a strictly positive matrix from $M(\bC^d)$. Then the following Cholesky decomposition theorem is known.

\begin{theorem}
Cholesky decomposition theorem
\begin{itemize}
\item[(1)] Let $A$ be a strictly positive matrix from $M(\bC^d)$. Then there exists a uniquely defined lower triangular matrix $L \in T_{l} (\bC^d)$ and such that the following equality holds:
\begin{equation}
A = L \cdot L^{\dagger} .
\end{equation}
\item[(2)] If $A$ is only positive semi-definite, then the decomposition 
\begin{displaymath}
A = L \cdot L^{\dagger},
\end{displaymath}
is still valid but the uniqueness statement on $L$ is no longer true in general.
\end{itemize}
\end{theorem}

Let $P(\bC^d)$ stand for a cone of strictly positive $d \times d$ matrices. Then we define the following map:

\begin{displaymath}
Ch: P(\bC^d) \rightarrow T_{l}(\bC^d) .
\end{displaymath}

The Cholesky map:
\begin{displaymath}
\begin{array}{lcl}
Ch   & : & A \longrightarrow L_A , \\
s.t. & : & L_A \cdot L_{A}^{\dagger} = A .
\end{array}
\end{displaymath}

Elementary properties of $Ch$:
\begin{itemize}
\item[(1)] $Ch(\alpha A) = \sqrt{\alpha} Ch(A), \alpha > 0$,
\item[(2)] $Ch(A + B) \neq Ch(A) + Ch(B)$ in general,
\item[(3)] if $A_n \rightarrow A$ in the norm and $A \in P(\bC^d)$, then $Ch(A_n) \rightarrow Ch(A)$ in the norm.
\end{itemize}

\begin{example}
Let $\Psi = \sum_{i=1}^{d} \psi_i e_i \in \bC^d$. Then $E_\psi = \mdenket{\psi}{\psi}$ has the matrix elements:
\begin{displaymath}
(E_{\psi})_{ij} = \overline{\psi_j} \psi_i .
\end{displaymath}

The Cholesky decomposition of the pure density matrix $E_\psi$ is given by:
\begin{displaymath}
A_{\Psi} = \left[
\begin{array}{ccc}
0 & \ldots & 0 \\
\vdots & \ddots & \vdots \\
0 & & 0 \\
\psi_1 & \ldots & \psi_d
\end{array}
\right] ,
\end{displaymath}
i.e.
\begin{displaymath}
E_\Psi = A^{\dagger}_{\Psi} \cdot A_{\Psi} .
\end{displaymath}
\end{example}

\begin{example}
If $\Psi \in \bC^d \otimes \bC^{d^{'}}$ and $\Psi$ is separable, i.e.: $\Psi = \Psi_1 \otimes \Psi_2$, $\Psi_i \in \bC^{d_i}$, then:
\begin{displaymath}
%(E_{\Psi})_{ij} = \mbraket{\Psi}{\Psi_i} \;\;\; \mathrm{where} \;\;\; \Psi_i = \mbraket{\Psi}{E_i} =  \mbraket{\Psi}{e_{\alpha(i)} \otimes f_{\beta(i)}}
Ch(E_{\Psi}) =Ch(E_{\Psi_1} \otimes E_{\Psi_2}) =  Ch(E_{\Psi_1}) \otimes Ch(E_{\Psi_2}).
\end{displaymath}
\end{example}
\begin{proof}
Proof by straightforward computations with the use of Lemma~\ref{lbl:lemma:tensor:products}.
\end{proof}

Having in mind possible applications of the Cholesky map Ch to a realistic physical situations one has to extend it to positive semi-definite matrices case as well. To this end, let us consider a $d \times d$ matrix $A \geq 0$, and such that: 
\begin{displaymath}
\dim( Ker(A)) = \dim ( \{ v \in \bC^d: A v =0 \} ) = k > 0.
\end{displaymath}

Then we can decompose
\begin{displaymath}
\bC^d = \mathrm{Ker}(A) \oplus (\mathrm{Ker}(A))^{\bot},
\end{displaymath}
where $\oplus$ means the direct product and $(\mathrm{Ker}(A))^{\bot}$ is the orthogonal complement of the kernel, $\mathrm{Ker}(A)$ of $A$. Let
\begin{equation}
A^{\downarrow} = A \upharpoonright \mathrm{Ker}(A)^{\bot},
\end{equation}
where the symbol $\upharpoonright$ means the corresponding restriction to a smaller domain. The restricted matrix $A^{\downarrow}$ is strictly positive and therefore by the use of the Cholesky theorem it follows that there exists a unique lower triangular matrix $L^{\downarrow}$ acting in d-k dimension and such that
\begin{displaymath}
A^{\downarrow} = L^{\downarrow} \cdot (L^{\downarrow})^{\dagger} .
\end{displaymath}

According to the decomposition $\bC^d = \mathrm{Ker}(A) \oplus (\mathrm{Ker}(A))^{\bot}$ we have the decomposition:
\begin{displaymath}
A = \left[
\begin{array}{cc}
A^{\downarrow} & 0 \\
0 & 0_{k} \\
\end{array}
\right] .
\end{displaymath}

Thus definining
\begin{displaymath}
L_o = \left[
\begin{array}{cc}
L^{\downarrow} & 0 \\
0 & 0 \\
\end{array}
\right],
\end{displaymath}
we have $L \in T_e(\bC^d)$ and moreover,
\begin{displaymath}
A = L_0 \cdot (L_0)^{\dagger} .
\end{displaymath}

The triangular matrix $L_0$, called a zero-extension of $L^{\downarrow}$, is one of many other possible extensions.

\begin{remark}
Let $(B_n)_n$ be a sequence of strictly positive $d \times d$ matrices and such that $B_n \rightarrow 0$ as $n \rightarrow \infty$ and in operator norm $|| \cdot ||$. Then for any $n$, the matrix:
\begin{displaymath}
A_n = A + B_n > 0 ,
\end{displaymath}
and therefore, again by the application of the Cholesky decomposition theorem, for any $n$ there exists a unique lower triangular matrix $L_n$ and such that
\begin{equation}
A_n = L_n \cdot (L_n)^{\dagger} .
\label{lbl:eq:A:decomposition}
\end{equation}
As $A_n \rightarrow A$ in operator norm, it follows from from Eq.~(\ref{lbl:eq:A:decomposition}) that also $\lim_{n \to \infty} L_n = L_{\infty}$ does exist and is lower triangular. From Eq.~(\ref{lbl:eq:A:decomposition}) with the $\lim_{n \to \infty}$ follows
\begin{equation}
A = \lim_{n \to \infty} A_n = \lim_{n \to \infty} (L_n \cdot L_{n}^{\dagger}) =  L_{\infty} \cdot L_{\infty}^{\dagger}.
\end{equation}
\end{remark}

Summarising:

\begin{observation}
Let $A$ be a positive semi-definite matrix on $\bC^{d}$ with non-trivial kernel $\mathrm{Ker}(A)$ of dimension $k > 0$. Then for any norm convergent to a zero sequence $(B_n)$ of strictly positive matrices $B_n$ there exists a lower triangular matrix $L$, such that the Cholesky decomposition is given as:
\begin{displaymath}
A = L \cdot L^{\dagger} ,
\end{displaymath}
and $L$ depends in general on the sequence $(B_n)$ chosen.
\label{lbl:obs:matrix:extension}
\end{observation}

\begin{definition}
An extension of the map $Ch$ to the positive semi-definite matrices set by the zero-extension method as outlined above will be denoted as $Ch_0$.
\end{definition}

Without further mention, we will always choose zero-extension decomposition in the Cholesky decompositions.

So, let $Q \geq 0$ and let $L$ be a lower triangular matrix obtained by the Cholesky decomposition, i.e.:
\begin{displaymath}
Q = L \cdot L^{\dagger} .
\end{displaymath}
We connect the following d-frame $F(L)$ with the matrix $L$ in the space $\bC^{d}$:
\begin{eqnarray}
F(L) = (r_1(L), \ldots, r_d(L)) \in dF(\bC^d) ,
\end{eqnarray}
where $r_{\alpha}(L)$ is the $\alpha$-th raw of $L$. Using the d-frame $F(L)$ we can construct the following vector
\begin{equation}
\Psi(Q) = \sum_{\alpha=1}^{d} e_{\alpha} \otimes r_{\alpha}(L) \in \bC^d \otimes \bC^d .
\label{lbl:eq:Psi:Q:decomposition}
\end{equation}

\begin{theorem}
Let $\rho \in E(\bC^d)$, i.e. $\rho \geq 0$ and $\mtr{\rho}=1$. Then there exists at least one vector $\Psi \in \bC^d \otimes \bC^d$ such that:
\begin{equation}
\Delta^{R}(\Psi) = \rho .
\end{equation}
\end{theorem}

\begin{proof}
First, we construct the vector
\begin{equation}
\Psi = \sum_{\alpha=1}^{d} e_{\alpha} \otimes r_{\alpha}(L),
\end{equation}
as described in formula~(\ref{lbl:eq:Psi:Q:decomposition}) and for $\rho$. Then, by the very definition of the right Gram operator $\Delta^R$, we compute easily:
\begin{equation}
\Delta^{R}(\Psi)_{\alpha\beta} = \rho_{\alpha\beta}, 
\end{equation}
where $\rho_{\alpha\beta} = \mbraket{e_{\beta}}{\rho \cdot {{e}_{\alpha}}}$.
\end{proof}

\begin{remark}
If   $\dim (\ker(\rho)) > 1$ then  there  exists   infinitely   many  different  extensions  of  the  Cholesky  map (see  Observation~\ref{lbl:obs:matrix:extension}), but in the following we always will use the canonical extension $Ch_{0}$.
\end{remark}

Thus we have constructed a map:
\begin{equation}
P_R : E(\bC^d) \rightarrow \partial E(\bC^d \otimes \bC^{d}),
\end{equation}
that we call a non-linear purification map: %The basic property of the map $P$ is that:
\begin{displaymath}
\Delta^{R}(P_{R}(\rho)) = \rho .
\end{displaymath}

The same construction applies to the left Gram operators. There exists a map:
\begin{displaymath}
P_L : E(\bC^d) \rightarrow \partial E(\bC^d \otimes \bC^d),
\end{displaymath}
and such that:
\begin{displaymath}
\Delta^{L}(P_{L}(\rho)) = \rho .
\end{displaymath}

\begin{remark}
Let $\Psi \in \bC^{d_1} \otimes \bC^{d_2}$ and let $\rho_\Psi$ be the corresponding density matrix $\rho_\Psi = \mdenket{\Psi}{\Psi} \in E( \bC^{d_1} \otimes \bC^{d_2} )$. Then, applying the purification map $P$ to 
$\rho_{\Psi}$:
\begin{equation}
P( \rho_{\Psi} ) = \mdenket{\overline{\Psi}}{\overline{\Psi}},
\end{equation}
for some $\overline{\Psi} \in (\bC^{d_1} \otimes \bC^{d_2}) \otimes (\bC^{d_1} \otimes \bC^{d_2} )$ with $|| \overline{\Psi} || = 1$. If $L(\Psi)$ is the corresponding lower triangular matrix for purification $\rho_{\Psi}$, then, if
\begin{equation}
\overline{\Psi} = \sum_{\alpha=1:d_1 d_2} e^{\otimes}_{\alpha} \otimes \overline{\psi}^{R}_{\alpha},
\end{equation}
then
\begin{equation}
(\overline{\psi}^{R}_{\alpha})_{i} = L (\Psi)_{\alpha i},
\end{equation}
for $i=1:d_1 d_2$.
\end{remark}

\begin{remark}
For an arbitrary $\rho \in \bC^{d_1} \otimes \bC^{d_2}$ seems to be a real challenge to read off the non-local properties of $\rho$ from the corresponding pure state $\overline{\Psi}_{\rho} \in (\bC^{d_1} \otimes \bC^{d_2}) \otimes (\bC^{d_1} \otimes \bC^{d_2} )$ obtained by the application of the purification map $P$ constructed here.
\end{remark}

\section{Geometrical aspects of entanglement} \label{lbl:sec:geo:of:entanglement}

Let us start with the following observation.

\begin{proposition}
The two-qudit pure state
\begin{equation}
\Psi = \sum_{i,j=1}^{d_1, d_2} \psi_{ij} e_i \otimes f_j \in \bC^{d_1} \otimes \bC^{d_2} \; \mathrm{with} \; d_1 \leq d_2,
\end{equation}
is a maximally entangled state if and only if the corresponding Gram operator $\Delta^{R}(\Psi)$ has the following form in some (and therefore in all) orthonormal bases $\{ g_i \}$:
\begin{equation}
\Delta^{R}(\Psi)_{i_1 i_2} = \frac{1}{d_1} \delta_{i_1 i_2} \;\;\; \mathrm{for} \;\;\; i_1 , i_2 = 1:d_1 .
\label{lbl:eq:gram:operator:form}
\end{equation}
\end{proposition}
\begin{proof}
If the Gram operator $\Delta^{R}(\Psi)$ possesses the form given by Eq.~(\ref{lbl:eq:gram:operator:form}) and there exists a some orthonormal basis $\{ g_i \}$, then the Schmidt numbers of the state $\mdenket{\psi}{\psi}$ are all equal to $\frac{1}{\sqrt{d_1}}$ and the corresponding entropy of entanglement is equal to $\log(d_1)$.

It is well-known that a commutant set of the unitary group $\bU(d_1)$ in the group $GL(d,\bC)$ is trivial and consists of multiplicities of the unity matrix  $E_{d_1}$ only. Therefore, if in some orthonormal basis $\Delta^{R}(\Psi)$ has the representation (\ref{lbl:eq:gram:operator:form}), then it has exactly the same representation (\ref{lbl:eq:gram:operator:form}) in any other orthonormal basis as the group $\bU(d_1)$ acts transitively on the manifold of all complete orthonormal frames of the space $\bC^{d_1}$. 
\end{proof}

Let $F = (f_1, \ldots, f_d)$ be a d-frame in $\bC^d$ consisting of linearly independent vectors $f_i$ and let $\Delta(F)$ be the Gram matrix built on $F$, i.e. $\Delta(F)_{\alpha\beta} = \mbraket{f_\beta}{f_{\alpha}}$. Let $\pol(\Psi)$ be a parallelepiped  constructed on the vectors $f_i$ composing the frame $F$.

\begin{lemma}
Let $F$, $\Delta(F)$ and $\pol(F)$ be as above. Then
\begin{equation}
\sqrt{ \det(\Delta(F))} = \mathrm{volume}_d(\pol(F)),
\end{equation}
where $\mathrm{volume}_d$ stands for the standard d-dimensional Euclidean volume of $\pol(F)$.
\end{lemma}

In what follows the $\mathrm{volume}_d$ will be denoted as $\mathrm{vol}_d$.

\begin{definition}
For any Gram matrix built on the d-frame $F = (f_1, \ldots, f_d)$ in $\bC^d$ and equipped with a standard euclidean scalar product we define a gramian of $F$ as G(F):
\begin{equation}
G(F) = \det( \Delta(F)) .
\end{equation}
\end{definition}

\begin{lemma} \label{lbl:lem:prop:gramians} Elementary properties of gramians, see i.e. \cite{Horn2013}.

\begin{itemize}

\item[(1)]For any d-frame $F = (f_1, \ldots, f_d)$ in $\bC^d$:
\begin{equation}
G(F) \geq 0,
\end{equation}
$G(F)=0$ is valid  iff the vectors $f_i$ forming the frame $F$ are linearly dependent.

\item[(2)] Let $F'$ be a d-frame obtained from $F$ by any permutations of the vectors composing $F$. Then
\begin{equation}
G(F') =G(F).
\end{equation}

\item[(3)] Let $F=F_1 \vee F_2$ be decomposed into two non-trivial frames $F_1$ and $F_2$. Then
\begin{equation}
G(F) \leq G(F_1) \cdot G(F_2) .
\end{equation}
The equality holds iff the subspaces generated by $F_1$ and $F_2$ are orthogonal to each other or one of gramians $G(F_i)=0$.

\item[(4)] For any d-frame $F = (f_1, \ldots, f_d)$ in $\bC^d$ and any $i=1:d$:
\begin{equation}
G(F) = G(F^i) \cdot h^2_i,
\end{equation}
where  
\begin{equation}
F^i = F \setminus \{ f_i \} \;\;\; \mathrm{and} \;\;\; h_i = \min_{F^i} \left| \left| f_i - \sum_{j \neq i}^d x^j f_j \right| \right| .
\end{equation}
\end{itemize}
\end{lemma}

Some conclusions:

\begin{conclusion}
Let for some d-frame $F$, $G(F)=0$. Then there exists a principal minor of $\Delta(F)$ which has a determinant equal to zero.
\label{lbl:con:minor:for:dframe}
\end{conclusion}

\begin{remark}
All the principal minors of $\Delta(F)$ are again Gram matrices. Therefore the parallelepiped corresponding to the principal minor as in Conclusion~\ref{lbl:con:minor:for:dframe}  must be a degenerated one.
\end{remark}

\begin{conclusion}
Iterating construction (4) given in Lemma~\ref{lbl:lem:prop:gramians} we conclude that for any sequence of indices $(i_1, \ldots, i_p)$, $p < d$, $1 \leq i_{\alpha} < d$ and $i_{\alpha} \neq i_{\alpha'}$ for $\alpha \neq \alpha'$ the following recurrence is valid:
\begin{eqnarray}
G(F^{s(0)}) & = & G(F), \notag \\
G(F^{s(1)}) & = & G(F^{s(0)}) \cdot h^2_{s(0)} , \notag \\ 
\ldots & \ldots & \ldots \notag \\
G(F^{s(\alpha+1)}) & = & G(F^{s(\alpha)}) \cdot h^2_{s(\alpha)} ,
\end{eqnarray}
where
\begin{eqnarray}
F^{s(0)} & = & F, \notag \\
F^{s(1)} & = & F^{s(0)} \setminus \{ f_{i_0}\}, \notag \\
\ldots & \ldots & \ldots \notag \\
F^{s(\alpha+1)} & = & F^{s(\alpha)} \setminus \{ f_{i_\alpha}\}, \notag \\
h_{s({\alpha})} & = & \min_{F^{s(\alpha)}} \left| \left| f_{i_\alpha} - \sum_{k \notin F(s(\alpha)) } x^{i_{\alpha'}} f_{i_{\alpha'}} \right| \right| .
\end{eqnarray}
\end{conclusion}

\begin{remark}
In the present context it is also worth mentioning the Hadamard inequality. For this, let $A \in M(\bC^d)$ with complex entries $A_{ij}$. The matrix $A$ can be seen as composed of d-vectors $r_i(A)=(A_{i_1},\ldots, A_{i_d})$. By the Hadamard inequality we have:
\begin{displaymath}
{| \det(A) |}^2 \leq \prod_{i=1}^{d} {|| r_i(A) ||}^2,
\end{displaymath}
and the equality holds true iff the system $(r_1(A), \ldots, r_d(A))$ is an orthonormal system of vectors or one of $r_i(A)$ is a zero vector. Geometrically,  according to the Hadamard inequality the volume of a parallelepiped built on $F_{A}=(r_1(A), \ldots, r_d(A))$ is never larger than the products of the length of its sides.
\end{remark}

\begin{remark}
Let $F$ be a k-frame (with $1 < k \leq d$) formed by $g_1, \ldots, g_k \in \bC^d$ and let $\Delta(F)$ be the corresponding $k \times k$ Gram matrix formed on $F$. Let $C(F)$ be the rectangular $d \times k$ matrix built from $g_i$ as columns. Then:
\begin{equation}
\Delta(F) = C(F)^{\dagger} \cdot C(F),
\end{equation}
is valid.
\end{remark}

\begin{lemma}
Let $A$ be a linear map $A: \bC^d \rightarrow \bC^d$ and $F = (f_1, \ldots, f_d)$ be a d-frame. Then
\begin{itemize}
\item[(i)] $\Delta(A(F)) = A \Delta(F) A^{\dagger}$,
\item[(ii)] $G(A(F)) = {\det(A)}^2 \cdot G(F)$.
\end{itemize}
\end{lemma}

\begin{proof}
By straightforward computations.
\end{proof}

Let $\Psi \in \bC^{d_1} \otimes \bC^{d_2}$, $|| \Psi || = 1$ and it is assumed that $d_1 = d_2 = d$ and moreover:
\begin{equation}
\rank{\Delta^{R}(\Psi)} = \rank{\Delta^{L}(\Psi)} = d.
\end{equation}
If $\Psi = \sum_{\alpha = 1}^{d} e_{\alpha} \otimes R_{\alpha}(\Psi)$ then we build a parallelepiped  $\pol(\Psi)$ on the frame $RF(\Psi) = \{ R_1(\Psi), \ldots, R_d(\Psi) \}$. The volume $\mathrm{vol}_d$ of $\pol(\Psi)$ is given as
\begin{equation}
vol_d ( \pol (\Psi) ) = \sqrt{ G(RF(\Psi)) } = \sqrt{ \det \Delta_R (\Psi)} .
% = \sqrt{ \bG(F^{'}(RF(\Psi^{'}))) }
\end{equation}

This volume is:
\begin{itemize}
\item[(i)] ${\bI}_d \otimes S\bU(d)$ -- invariant,
\item[(ii)] if $T \in \myendfun(\bC^d)$ is such that $|| T || \leq 1$, then $\mathrm{vol}_d( \pol( (1 \otimes T) \Psi ) ) \leq \mathrm{vol}_d ( \pol( \Psi ) ) $,
\item[(iii)] $\sup_{\Psi} \mathrm{vol}_{d} ( \pol(\Psi) ) = \Psi^{\star}$ where $\Psi^{\star}=\frac{1}{\sqrt{d}} \sum_{\alpha=1}^{d} e_i \otimes f_i$.
\end{itemize}

Let us define the following map for $\Psi \in \bC^d \otimes \bC^d$:
\begin{equation}
\begin{array}{lll}
\gen(\Psi) & = & \sqrt[d]{ \det ( \Delta^L (\Psi) \otimes \Delta^R (\Psi) ) }  \\
		  & = & \sqrt[d]{ (\det ( \Delta^L (\Psi) ))^d \cdot (\det ( \Delta^R (\Psi) ))^d } \\
		  & = & \det ( \Delta^L (\Psi)) \det(\Delta^R (\Psi) ) \\
		  & = & \mathrm{vol}^4_d ( \pol ( RF(\Psi) ) ) .
\end{array}
\end{equation}
Then:
\begin{itemize}
\item[(i)] $\gen$ is $S\bU(d) \otimes S\bU(d)$ invariant on $\partial E^{\star} ( \bC^d \otimes \bC^d ) = \{ \Psi \in \bC^d \otimes \bC^d  \; \mathrm{and} \; || \Psi || = 1\}$ as above, i.e. 
\begin{equation}
\gen(\Psi) = \gen ((U_1 \otimes U_2) \Psi) \;\;\; \mathrm{for} \;\;\; U_1 \otimes U_2 \in S\bU(d) \otimes S\bU(d).
\end{equation}
\item[(ii)] if $T_i \in \myendfun(\bC^d)$, $|| T_i || \leq 1$, $i=1:2$, then for any $\Psi \in \partial E^{\star}(\bC^d \otimes \bC^d)$:
\begin{equation}
\gen(\Psi) \leq  \gen ((T_1 \otimes T_2) \Psi) .
\end{equation}
\item[(iii)] $\sup_{\Psi} \gen(\Psi) = \Psi^{\star}$, for $\Psi \in \partial E^{\star}( \bC^d \otimes \bC^{d} )$ where $\Psi^{\star}=\frac{1}{\sqrt{d}}\sum_{\alpha=1}^{d} e_1 \otimes f_i$.
\end{itemize}

\begin{remark}
The isoperimetrical problem: having a d-frame $V=(v_1, \ldots, v_d)$ in $\bC^d$ with the constraint: ${|| v_1 ||}^2 + \ldots + {|| v_d ||}^2 = 1$ to find the parallelepiped constructed on $V$ denoted as $\pol(V)$ and such that:
\begin{equation}
\sup_{V} \mathrm{vol}_d (P(V)) = V^{\star}.
\end{equation}
It is widely and well-known that it has a unique solution $V^{\star}=(v^{\star}_1, \ldots, v^{\star}_d)$ such that $\mbraket{v^{\star}_i}{v^{\star}_{j}}=\delta_{ij} \cdot \frac{1}{{d}}$.
\end{remark}

To proceed further let, us define the following Schmidt foliation of the set $\partial E( \bC^{d_1} \otimes \bC^{d_2} )$, $d_1 \leq d_2$:
\begin{equation}
\partial E ( \bC^{d_1} \otimes \bC^{d_2} ) = U^{d_1}_{k \geq 0} \partial_{k} E ( \bC^{d_1} \otimes \bC^{d_2}),
\end{equation}
where
\begin{equation}
\partial_{k} E ( \bC^{d_1} \otimes \bC^{d_2}) = \{ \Psi \in \bC^{d_1} \otimes \bC^{d_2} : \mathrm{Schmidt \; rank \; of} \; \Psi=k \; \mathrm{and} \; || \Psi || = 1 \},
\end{equation}
for $k > 0$, and 
\begin{equation}
\partial_{0} E ( \bC^{d_1} \otimes \bC^{d_2}) = \emptyset .
\end{equation}

For any $\Psi \in \partial_k E (\bC^{d_1} \otimes \bC^{d_2})$, $0 < k \leq d_1$ it follows that
\begin{equation}
\rank{ \Delta_{R}(\Psi) } = \rank{ \Delta_{L}(\Psi) } = k,
\end{equation}
and if $\rank{ \Delta_{R}(\Psi) }=k$, then $\Psi \in \partial_k E (\bC^{d_1} \otimes \bC^{d_2})$.

If $\Psi \in \partial_k E (\bC^{d_1} \otimes \bC^{d_2})$, $0 < k \leq d_1 < d_2$ and let $FR_k(\Psi)$ be a subframe of $FR(\Psi)$ obtained by choosing a maximal subset of $FR(\Psi)$ of k-linearly independent vectors from $FR(\Psi)$. Similarly, we define a restricted left k-subframe $LF_k(\Psi)$ of $LF(\Psi)$. The corresponding Gram matrices $\Delta^R_k ( \Psi )$, $\Delta^L_k ( \Psi )$ and $\Delta_k ( \Psi ) = \Delta^L_k ( \Psi ) \otimes \Delta^R_k ( \Psi )$ are then constructed on the restricted k-subframes $RF_k(\Psi)$, resp. $LF_k(\Psi)$.

\begin{definition}
The gramian volume map
\begin{equation}
\gen : \partial E( \bC^{d_1} \otimes \bC^{d_2}) \rightarrow [0,1],
\label{lbl:eq:gram:vol:map}
\end{equation}
where $d_1 \leq d_2$ is defined by the following
\begin{itemize}
\item[(1)] if $\Psi \in \partial_0 E( \bC^{d_1} \otimes \bC^{d_2}) \cup \partial_1 E( \bC^{d_1} \otimes \bC^{d_2})$ then $\gen(\Psi)=0$,
\item[(2)] if $\Psi \in \partial_k E( \bC^{d_1} \otimes \bC^{d_2})$ for $1 < k \leq d_1$ then 
\begin{equation}
\gen(\Psi)=\sqrt[k]{ \det ( \Delta^L_k(\Psi) \otimes  \Delta^R_k(\Psi))  } .
\end{equation}
\end{itemize}

\end{definition}

\begin{proposition}
The gramian volume map
\begin{equation}
\gen : \partial E( \bC^{d_1} \otimes \bC^{d_2}) \rightarrow [0,1], d_1 \leq d_2,
\end{equation}
has the following properties for any $1 < k \leq d_1$:
\begin{itemize}
\item[(i)] $\gen_{\upharpoonright \partial_k E(\bC^{d_1} \otimes \bC^{d_2})}$ is $S\bU(k) \otimes S\bU(k)$ invariant,

\item[(ii)] If $T_1, T_2 \in \myendfun(\bC^k)$, then
\begin{gather}
\gen_{\upharpoonright \partial_k E(\bC^{d_1} \otimes \bC^{d_2})} ( (T_1 \otimes T_2)(\Psi) ) = \notag \\
\det(T_1 \otimes T_2) ( \gen_{\upharpoonright \partial_k E(\bC^{d_1} \otimes \bC^{d_2})} (\Psi) ) \det(T^{\dagger}_1 \otimes T^{\dagger}_2) .
\end{gather}

\item[(iii)] Let $\sup_{\Psi} \gen_{\upharpoonright \partial_k E(\bC^{d_1} \otimes \bC^{d_2})} (\Psi) = \Psi^{\star}_{k}$, then the Schmidt decomposition of $\Psi^{\star}_{k}$ is given as
\begin{equation}
\Psi^{\star}_{k} = \frac{1}{\sqrt{k}} \sum_{i=1}^{k} \overline{e}_i \otimes \overline{f}_i,
\end{equation}
where $\{ \overline{e}_i \}, \{ \overline{j}_i \} \in \mathrm{CONS}(\bC^k)$.

\item[(iv)] If $\Psi, \Psi^{'} \in \partial_{k} E(\bC^{d_1} \otimes \bC^{d_2})$ and
\begin{equation}
\begin{array}{lcl}
\underline{\lambda}_{\Psi} & = & (\lambda^1_{\Psi}, \ldots, \lambda^k_{\Psi}), \\
\underline{\lambda}^{'}_{\Psi} & = & (\lambda^{{'}_1}_{\Psi}, \ldots, \lambda^{{'}_k}_{\Psi}),
\end{array}
\end{equation}
are the corresponding non-zero Schmidt coefficients of $\Psi$, resp. of $\Psi^{'}$ and there exists a permutation $\pi \in S_k$ such that for all $i=1:k$
\begin{equation}
\lambda^{i}_{\Psi} \leq \lambda^{' \pi(i)}_{\Psi},
\end{equation}
then $\gen(\Psi) \leq \gen(\Psi')$.

\item[(v)] If $T_1 \otimes T_2 \in SL(2, \bC)$, then $\gen((T_1 \otimes T_2)(\Psi))=\gen(\Psi)$.

\end{itemize}
\end{proposition}

\begin{proof}
All claims follow easily from the facts established previously. Details are left for the potential reader.
\end{proof}

\begin{remark}
The question whether the Gramian volume could be used as a quantitative measure of entanglement (properly  normalized)  given  by (\ref{lbl:eq:gram:vol:map}) in the  commonly  accepted  sense  \cite{Plenio2007, Bengtsson2016, Horodecki2009} is  discussed in  \cite{RGMSPrepar2020} in a more details.
\end{remark}

\section{Conclusions} \label{lbl:sec:conclusions}

A systematic Gram matrix-based analysis of quantum entanglement in bipartite qudits systems has been carried out. Some applications of the new mathematical technique invented are presented. One of the applications includes a construction of a certain purification map based on the Cholesky decomposition. An interesting example would be the use of the proposed here purification approach for existing examples of entanglement purification, e.g. \cite{Behera2019}. The other main result presented here is the explicit comparison of the amount of entanglement included in a given two-qudit state to the d-dimensional, euclidean volume of a certain parallelepiped connected with the state under analysis. The introduced volume is locally SU invariant and non-increasing under general local physical operations.

Another applications of the Gram matrix theory  in the  present  context  contain    the analysis  of  the  bipartite  systems  of  the  type  $(d, \infty)$, where  $d <  \infty$ \cite{RGMW2020} and also genuine  infinite-dimensional bipartite  systems \cite{RG2020a}.

\section*{Acknowledgements} \label{lbl:sec:acknowledgements:RG:MS:2020}
\noindent

We would like to thank for useful discussions with the~\textit{Q-INFO} group at the Institute of Control and Computation Engineering (ISSI) of the University of Zielona G\'ora, Poland.

\appendix

\section{Basic mathematical notions}\label{lbl:app:math:basic}

In this section, we recall basic definitions and notions used in the paper. The definitions given here are probably well-known to the reader, however, for the clarity of the presentation, all basic definitions and conventions from the theory of linear algebra related to the gramian notion are collected in this appendix.

\subsection{Matrices and vectors}

The algebra of $n \times m$ size matrices over the complex numbers $\bC$ (or real numbers $\bR$) is denoted as $M_{n \times m}(\bC)$  (resp. $M_{n \times m}(\bR))$ and $M_n$ stands for the corresponding square analogue. We skip sometimes the field if the matrix can be either real or complex without changing the result.

The $( i,j)$-th entry of a matrix $M \in M_{n \times m}$ is referred to by $(M)_{ij}$ or by $m_{ij}$. Let $A$ be a matrix. Then we denote by $A^T$ its transpose, by $A^*$ its (complex) conjugate, by $A^{\dagger}$ its conjugate transpose, by $A^{-1}$ its inverse (if it exists, i.e. $A$ is nonsingular) and by $\det(A)$ we denote its determinant.

Furthermore, we introduce the following special vectors and matrices. Let $E_n$ (or $\bI_n$) be an identity matrix of dimension $n$. The dimension is omitted if it is clear from the context. $E_{ij}$ is the $ij$-th elementary $0-1$ projection matrix where zero value is in each position except the value equal to one in position (i,j).

The following inner product, called a Hilbert-Schmidt product, introduces a Hilbert space structure in the space of matrices:
\begin{equation}
\mbraket{ M }{ N } =\mtr{M^{\dagger}N} ,
\end{equation}
for $M,N \in M_n$ and where $\mtr{ \cdot }$  is the matrix trace, i.e.:
\begin{equation}
\mtr{M} = \sum_{i \in 1:n} M_{ii} .
\end{equation}
The corresponding norm
\begin{equation}
M = \mtr{M^{\dagger}M}^{\frac{1}{2}},
\end{equation}
is known as the Frobenius norm.

For any matrix $A \in M_n$ we introduce the $\mathrm{vec}$ operation, defined as:
\begin{equation}
\mathrm{vec}(A) = ( a_{11}, \ldots, a_{n1}, \ldots\ldots, a_{1n} \ldots, a_{nn})^{T} .
\end{equation}

Then the Hilbert-Schmidt inner product can be defined as:
\begin{equation}
\mtr{A^{\dagger} B} = \mbraket{\mathrm{vec}(A)}{\mathrm{vec}(B)}_{\bC^{n^2}}.
\end{equation}

A hermitian matrix $A \in M_n (\bC)$ is called positive semi-definite and denoted as $A \geq 0$ ) iff for any $v \in \bC^d$:
\begin{equation}
\mbraket{v}{Av} \geq 0.
\label{lbl:eq:semidefinite:of:A}
\end{equation}

A matrix $A$ is positive definite if the inequality (\ref{lbl:eq:semidefinite:of:A}) is strict for all non zero $v \in \bC^d$.

Recall that the eigenvalues of a square matrix $A \in M_{n \times n}(\bC)$ are the numbers $\lambda \in \bC$ that satisfy the eigenvalue equation $A{v_{\lambda}}=\lambda {v_{\lambda}}$ for some non-zero $v \in \bC^n$. The spectrum of $A$, which is the set of all eigenvalues, is denoted as $\sigma(A)$. The spectral decomposition of a normal matrix $A \in M_{n \times n}(\bC)$ is given by the formula

\begin{equation}
A = \sum_{\sigma(A)} \lambda E_{v_\lambda},
\end{equation}
where $E_{v_\lambda}$ is an orthogonal projector onto the vector $v_\lambda$.

The singular values of a matrix $A \in  M_{n \times m}$ are the square roots of the $\min(n, m)$ (counting multiplicities) largest eigenvalues of $A^+A$. The singular value decomposition of $A$ is given by the following formula:
\begin{equation}
A= VDW^{\dagger} ,
\label{eq:lbl:singular:value:decomposition}
\end{equation}
where $V \in M_n$, $W \in M_m$ are unitary and $D$ is a diagonal matrix containing singular values (ordered by non-increasing size of them) on the diagonal.

From the decomposition given by Eq.~(\ref{eq:lbl:singular:value:decomposition}) it follows that the rank $A$ is the number of its non-zero singular values.

\subsection{Tensor (Kronecker) product of matrices} \label{lbl:ssec:tensor:matrix}

The tensor product of two matrices $A \in M_{n \times m}$, $B \in M_{p \times q}$ is defined as
\begin{equation}
A \otimes B = \left( 
\begin{array}{ccc}
a_{11} B & \ldots & a_{1m}B \\
\vdots & \ddots & \vdots \\ 
a_{n1} B & \ldots & a_{nm}B 
\end{array}
\right) \in M_{pn \times qm} (\bC) .
\end{equation}

Notation concerning the tensor product:
\begin{itemize}
\item notation $\sum_{i=1:d_1}$ means the "summation over $i$ from $1$ to $d_1$",
\item notation $i_\alpha \in 1:d_1$ means "$i$ from the set $\{1, \ldots, d_1\}$".
\end{itemize}

A. For vectors:

Let $v \in \bC^{d_1}$, $w \in \bC^{d_2}$ then $v=\sum_{i=1:d_1} v_i e_i$, $w=\sum_{i=1:d_2} w_i f_i$, and we have:
\begin{equation}
v \otimes w = \sum_{i,j} v_i w_j e_i \otimes f_j = \sum_{\alpha=1:d_1 d_2} V_{\alpha} E_{\alpha}^{\otimes},
\end{equation}
where $(E_{\alpha}^{\otimes})_\beta = \delta_{\alpha\beta}$ is the canonical basis of the space $\bC^{d_1 d_2}$. For $\alpha \in 1:d_1 d_2$  we can decompose  $\alpha$ in a unique way:
\begin{equation}
\alpha = (i_\alpha - 1) d_2 + j_{\alpha}, \; \mathrm{where} \; i_\alpha \in 1:d_1 \; \mathrm{and} \; j_{\alpha} \in 1 : d_2 \; \mathrm{and} \; \mathrm{then} \; V_{\alpha} = v_{i_{\alpha}} w_{j_{\alpha}}.
\label{lbl:eq:alpha:i:index}
\end{equation}
In particular, we have: $e_i \otimes f_j = E^{\otimes}_{\alpha(i,j)}$, where:
\begin{equation}
\alpha(i,j) = (i-1) d_2 + j, \; i \in 1:d_1, \; j \in 1:d_2 .
\label{lbl:eq:alpha:ij:index}
\end{equation}

B. for matrices –local versus global canonical bases:

Let $(E_{ij}, i,j=1:d)$ be a system of well-known $0-1$ matrices forming a canonical basis in the space $M_{d \times d}(\bC)$, i.e.:
\begin{equation}
{(E_{ij})}_{\alpha\beta} = \delta_{\alpha i} \delta_{j \beta} \; \mathrm{for} \; \alpha, \beta = 1:d .
\end{equation}

If $A=(a_{ij}) \in M_{n \times m}(\bC), B=(b_{ij}) \in M_{p \times q}(\bC)$, then the tensor product of $A$ and $B$ is given by
\begin{equation}
(A \otimes B)_{\alpha\beta} = a_{i_\alpha j_\alpha}b_{i_\beta j_\beta},
\end{equation}
where the indices $i_{\alpha}, \ldots, j_{\beta}$ have to be computed from the equations:
\begin{equation}
\begin{array}{lcl}
\alpha & = & (i_{\alpha} - 1)p + j_{\alpha}, \; \mathrm{where} \; i_{\alpha} \in 1:n \; \mathrm{and} \; j_{\alpha} \in 1:p , \\
\beta & = & (i_{\beta} - 1)q + j_{\beta}, \; \mathrm{where} \; i_{\beta} \in 1:m \; \mathrm{and} \; j_{\beta} \in 1:q .
\end{array}
\end{equation}

In particular, we have the following formula: if $E^1_{ij} \in M_{n \times m}(\bC)$, and $E^2_{kl} \in M_{p \times q}(\bC)$ are canonical bases  in  the corresponding space of matrices, then
\begin{gather}
E^1_{ij} \otimes E^2_{kl} = E^{\otimes}_{\alpha\beta} \; \mathrm{for} \notag \\
\alpha = (i - 1)p + j, \; \mathrm{where} \; i \in 1:n \; \mathrm{and} \; j \in 1:p , \notag \\
\beta = (i - 1)q + j, \; \mathrm{where} \; i \in 1:m \; \mathrm{and} \; j \in 1:q .
\end{gather}

\subsection{Some basic properties of the Kronecker product} \label{lbl:subapp:Kronecker:Product}

For the completeness of the paper and for the reader's convenience  we list some, more or less known albeit elementary, properties of the introduced tensor product of matrices.

\begin{itemize}
\item[(TP1)] For any $A \in M_{n \times m}$, $B \in M_{p \times q}$ the following formulas hold true:
	\begin{itemize}
		\item[(i)] ${(A \otimes B)}^{T} = A^T \otimes B^T$,
		\item[(ii)] ${(A \otimes B)}^{\star} = A^{\star} \otimes B^{\star}$,
		\item[(iii)] ${(A \otimes B)}^{\dagger} = A^{\dagger} \otimes B^{\dagger}$,
		\item[(iv)] if $A^{-1}$ and $B^{-1}$ exist then ${(A \otimes B)}^{-1}$ also exists and the equality ${(A \otimes B)}^{-1} = A^{-1} \otimes B^{-1}$ holds true.
	\end{itemize}
\item[(TP2)] For any $A \in M_{n}$, $B \in M_{p}$:
	\begin{itemize}
		\item[(i)] $\det(A \otimes B) = {\det(A)}^p {\det(B)}^n $,
		\item[(ii)] $\mtr{A \otimes B} = \mtr{A} \mtr{B}$.
	\end{itemize}
\item[(TP3)] For any $A \in M_{n}$, $B \in M_{p}$:
\begin{equation}
\sigma( A \otimes B) = \sigma(A) \cdot \sigma(B) \equiv \{ \lambda_A \lambda_B, \lambda_A \in \sigma(A), \lambda_B \in \sigma(B) \}.
\end{equation}
And then, if
\begin{equation}
(A \otimes B) \Psi_{\lambda_A \lambda_B} = \lambda_A \lambda_B \Psi_{\lambda_A \lambda_B},
\end{equation}
then $\Psi_{\lambda_A \lambda_B}=\Psi_{\lambda_A} \otimes \Psi_{\lambda_B}$, where $A\Psi_{\lambda_A}=\lambda_{A}\Psi_{\lambda_A}$ and $B\Psi_{\lambda_B}=\lambda_{B}\Psi_{\lambda_B}$ are the corresponding eigenfunctions.
\item[(TP4)] For any $A \in M_{n \times m}$, $B \in M_{p \times q}$ such that $\rank{A}=r_A$, $\rank{B}=r_B$ and for which the corresponding singular value decompositions are
\begin{equation}
A=V_A D_A W_A, \; B=V_B D_B W_B,
\end{equation}
(see for \cite{SVDProp1}) the following SVD formula holds true
\begin{equation}
A \otimes B = (V_A \otimes V_B) (D_A \otimes D_B) (W_A \otimes W_B).
\label{lbl:eq:svd:prop:1}
\end{equation}
In particular, it follows from Eq.~(\ref{lbl:eq:svd:prop:1}) that
\begin{itemize}
\item[(i)] $\rank{A \otimes B}=r_A r_B$,
\item[(ii)] $\sigma_{sv}(A \otimes B)$ = $\sigma_{sv}(A) \cdot \sigma_{sv}(B) \equiv \left\{ \mu_{A} \cdot \mu_{B}, \mu_{A} \in \sigma_{sv}(A), \mu_{B} \in \sigma_{sv}(B) \right\}$.
\end{itemize}

\end{itemize}


\begin{thebibliography}{000}
%%%\bibitem{first}
%%%P. Horodecki and R. Horodecki (2001), {\it Distillation and bound entanglement},
%%%Quantum Inf. Comput., Vol.1, pp. 045-075.
%%%
%%%\bibitem{cal}
%%%R. Calderbank and P. Shor (1996), {\it Good quantum error
%%%       correcting codes exist},
%%%Phys. Rev. A, 54, pp. 1098-1106.
%%%
%%%\bibitem{niel}
%%%M.A. Nielsen and J. Kempe (2001), {\it Separable states are
%%%more disordered globally than locally}, quant-ph/0105090.
%%%
%%%\bibitem{mar}
%%%A.W. Marshall and I. Olkin (1979), {\it Inequalities: theory of majorization and its applications},
%%%Academic Press (New York).


\bibitem{Schrodinger1935} 
E. Schr\"odinger (1935), {\it Die gegenw\"artige Situation der Quantenmechanik},
Naturwissenschaften, 23, 807 

\bibitem{Nielsen2000} 
M.A. Nielsen and I. L. Chuang (2000), {\it Quantum Computation and Quantum Information},
Cambridge University Press, Cambridge, U.K. 

\bibitem{Bengtsson2016} 
I. Bengtsson and K. \.Zyczkowski (2017), {\it Geometry of Quantum States: An Introduction to Quantum Entanglement},
2nd-edition, Cambridge University Press, Cambridge, U.K.

\bibitem{Horodecki2009}
R. Horodecki, P. Horodecki, M. Horodecki, K. Horodecki (2009), {\it Quantum Entanglement},
Rev. Mod. Phys., Vol. 81, Iss. 2, 865  available also at arXiv: quant-ph:/0702225.

\bibitem{Guhne2008}
O. Guhne, G. Toth (2008),
{\it Entanglement Detection},
Physics Reports, Vol. 474, Issue 1-6, pp. 1-75, available  at arXiv:quant-ps:/0811.2803v3.

\bibitem{Horodecki2010}
R. Horodecki, S.Ya. Kilin, J.S. Kowalik (edts.) (2010),
{\it Quantum Cryptography and Computing}, IOS Press 

\bibitem{Griffiths2004}
J. Griffiths (2004),
{\it Introduction to Quantum Mechanics}, (2nd edition). Prentice Hall.

%\bibitem{cite13}
%P. Blanchard, E. Bruning (2015) 
%Mathematical Methods in Physics, Birkhauser, Theorem 26.8 , p. 387

%\bibitem{cite14} Michael Reed, Barry Simon, (1980) Methods of Modern Mathematical Physics, Vol. 1: Functional Analysis. Berlin: Academic Press, pp. 50

\bibitem{RGielerak2017} R. Gielerak, (2017),
{\it Entangling Qubit with the rest of the world - the monogamy principle in action},
Studia Informatica, Vol. 38, No.~3, pp. 33--43, and in preparation.

%\bibitem{cite17} Gheorghiu, V., Griffiths, R.B. (2008) Phys. Rev. A 78, 020304(R)

%\bibitem{cite18} Nielsen M.A., Vidal, G. (2001) Majorisation and the interconversion of bipartite states, Quantum Information and Computation, Vol.~1, No.~1  pp.~76-93

%\bibitem{cite20} Cavalcanti D., Brandao, F.G.S.L. and Terra Cunha M.O., (2005) Are all maximally entangled states pure?, Phys. Rev. A72 040303(R) 

\bibitem{Plenio2007} M.B. Plenio, S. Virmani (2007),
{\it An introduction to entanglement measures},
Quantum Information \& Computation, Vol. 7 Issue 1  also available at:  arXiv:quant-ph/0504163v3 

\bibitem{Bourbaki1989}
N. Bourbaki (1989),
{\it Elements of Mathematics}, Algebra 1, Springer Verlag

\bibitem{MacLane1999}
S. MacLane, G. Birkhoff (1999),
{\it Algebra}, 3rd edition, American Mathematical Society

\bibitem{Spivak1965} 
M. Spivak (1965),
{\it Calculus on Manifolds}, Addison-Wesley

\bibitem{IBM16qEnt} Y. Wang, Y. Li, Z. Yin, et al. (2018),
{\it 16-qubit IBM universal quantum computer can be fully entangled},
npj Quantum Inf 4, 46 DOI:10.1038/s41534-018-0095-x

\bibitem{Hazenwinkel2001} M. Hazenwinkel, (ed.) (2001),
{\it Gram matrix},
In: Encyklopaedia of Mathematics, Springer Science+Business Media B.V.

\bibitem{Ritz2018} 
C. Ritz (2018),
{\it Characterizing the structure of multiparticle entanglement in high-dimensional systems}, Phd Thesis,

\bibitem{Caban2017} P. Caban, J. Rembieliński, K.A. Smoliński, et al. (2017),
{\it SLOCC orbit of rank-deficient two-qubit states: quantum entanglement, quantum discord and EPR steering}, Quantum Information Processing  Vol. 16, Issue 7 DOI:10.1007/s11128-017-1626-7

\bibitem{Bryan2019} J. Bryan, S. Leutheusser, Z. Reichstein, M. Van Raamsdonk (2019),
{\it Locally Maximally Entangled States of Multipart Quantum Systems}, Quantum, Vol. 3, pp.~115 

\bibitem{SVDProp1} G.H. Golub,  C. F. Van Loan (1996),
{\it Matrix Computations}, Johns Hopkins University Press; 3rd edition 

\bibitem{LieBook1}  
B. C. Hall (2003), {\it Lie Groups, Lie Algebras, and Representations: An Elementary Introduction},
Springer, New York, NY %DOI: 10.1007/978-0-387-21554-9

\bibitem{LieBook2} I. Stewart (1970), {\it Lie Algebras},
Springer, Berlin, Heidelberg 

\bibitem{Horn2013} 
R.A. Horn, C.R. Johnson (2013), 
{\it Matrix Analysis}, 2nd edition,
Cambridge University Press

\bibitem{RGMSPrepar2020} R. Gielerak, M.Sawerwain (2020), in preparations.

\bibitem{Behera2019}
B.K. Behera, S. Seth, A. Das, P.K. Panigrahi (2019), {\it Demonstration of entanglement purification and swapping protocol to design quantum repeater in IBM quantum computer}, Quantum Inf Process Vol.~18, No.~108  DOI:10.1007/s11128-019-2229-2

\bibitem{RGMW2020} R. Gielerak, M. Wroblewski (2020) Schmidt  decompositions of quantum states for $(d, \infty)$ systems, available also at arXiv:1803.09541

\bibitem{RG2020a} R. Gielerak (2020), On the use of Fredholm determinants for an analysis of entanglement in bipartite, infinite dimensional systems, in preparations.


\end{thebibliography}
\end{document}